\documentclass[showpacs,amsmath,amssymb,twocolumn,pra,superscriptaddress]{revtex4-1}

\usepackage[dvips]{graphicx} 
\usepackage{amsfonts}
\usepackage{amssymb}
\usepackage{amscd}
\usepackage{amsmath}    
\usepackage{enumerate}
\usepackage{epsfig}
\usepackage{subfigure}
\usepackage{xcolor}
\usepackage{amsthm}
\usepackage{framed}
\usepackage{multirow}

\newtheorem{theorem}{Theorem}

\newtheorem{observation}{Observation}

\newcommand{\bra}[1]{\mbox{$\left\langle #1 \right|$}}
\newcommand{\ket}[1]{\mbox{$\left| #1 \right\rangle$}}

\newcommand{\comments}[1]{}

\begin{document}
\preprint{APS/123-QED}
\title{Entanglement detection under coherent noise: Greenberger-Horne-Zeilinger-like states
}
\date{\today}
\author{You Zhou}
\email{zyqphy@gmail.com}
\affiliation{Center for Quantum Information, Institute for Interdisciplinary Information Sciences, Tsinghua University, Beijing 100084, China}

\begin{abstract}
Entanglement is an essential resource in many quantum information tasks and entanglement witness is a widely used tool for its detection.
In experiments the prepared state generally deviates from the target state due to some noise. Normally the white noise model is applied to quantifying such derivation and in the same time reveals the robustness of the witness. However, there may exist
other kind of noise, in which the coherent noise can dramatically ``rotate" the prepared state. In this way, the coherent noise is likely to lead to a failure of the detection, even though the underlying state is actually entangled. In this work, we propose an efficient entanglement detection protocol for $N$-partite Greenberger-Horne-Zeilinger (GHZ)-like states. The protocol can eliminate the effect of the coherent noise and in the same time feedback the corresponding noise parameters, which are beneficial to further improvements on the experiment system. In particular, we consider two experiment-relevant coherent noise models, one is from the unconscious phase accumulation on $N$ qubits, the other is from the rotation on the control qubit. The protocol effectively realizes a family of entanglement witnesses by postprocessing the measurement results from $N+2$ local measurement settings, which only adds one more setting than the original witness specialized for the GHZ state.
Moreover, by considering the trade-off between the detection efficiency and the white-noise robustness, we further reduce the number of local measurements to $3$ without altering the performance on the coherent noise.
Our protocol can enhance the entanglement detection under coherent noises and act as a benchmark for the state-of-the-art quantum devices.

%
%
%
%
%
%
%
\end{abstract}

\maketitle
\section{Introduction}
Entanglement, as a unique feature of quantum mechanics, plays an essential role in many quantum information processing tasks, such as quantum teleportation \cite{Bennett1993Teleporting}, quantum cryptography \cite{Bennett84cryptography,Ekert1991cryptography}, non-locality test \cite{Brunner2014nonlocality}, quantum computing \cite{Nielsen2011Quantum}, quantum simulation \cite{Lloyd1996Simulators} and quantum metrology \cite{Wineland1992squeezing,Giovannetti2004measure}.
Consequently, it is quite significant to detect entanglement in experimental systems, which not only acts as benchmark and calibration of the underlying platform, but also certifies useful quantum resources for the further information processing. So far, tremendous efforts have been devoted to the realization of multipartite entanglement in various systems \cite{Monz2011Entanglement,Britton2012trapped,Nigg2014encoded,Song2017Entanglement,Gong2019Genuine,Wang2016Entanglement,chen2017observation,12photon, Toth2014Detecting,Luo2017Deterministic,Lange2018atomic}. In particular, the genuine multipartite entanglement is witnessed in 14-ion-trap-qubit \cite{Monz2011Entanglement}, 10-superconducting-qubit \cite{Song2017Entanglement}, and 12-photon-qubit systems \cite{12photon}, with the target state being the Greenberger-Horne-Zeilinger (GHZ) state. 

The detection of genuine multipartite entanglement is generally a challenging task, since the dimension of the Hilbert space increases exponentially with respect to the system size. Compared with the unfeasible quantum state tomography \cite{Vogel1989Determination,Paris2004esimation}, the entanglement witness is an useful tool to realize it \cite{TERHAL2001witness,GUHNE2009detection}. The witness is usually a Hermitian operator $\mathcal{W}$, satisfying that $\mathrm{Tr}(\mathcal{W}\sigma_s)\geq 0$ for all separable states $\sigma_s \in \mathcal{S}_{sep}$, with $\mathcal{S}_{sep}$ the separable state set; $\mathrm{Tr}(\mathcal{W}\ket{\Psi}\bra{\Psi})<0$ for some entangled state $\ket{\Psi}$, such as the GHZ state. Consequently, if $W$ returns a negative value, one can confirm that the prepared state is entangled; a non-negative value tells nothing, denoted as a null result.

A straightforward way to construct a witness is based on the intuition that the prepared state $\rho_{pre}$ is entangled if it is close to an entangled target state, say $\ket{\Psi}$. To be specific,
\begin{equation}\label{Eq:witness}
\begin{aligned}
\mathcal{W}=\alpha \mathbb{I}-\ket{\Psi}\bra{\Psi},
\end{aligned}
\end{equation}
where $\alpha$ is the maximal fidelity between $\ket{\Psi}$ and all separable states $\sigma_s$, i.e., $\alpha=sup\{\bra{\Psi}\sigma_s\ket{\Psi}|\sigma_s \in S_{sep}\}$. On account of the convexity of $S_{sep}$, $\alpha$ can be determined by the maximal Schmit coefficient of $\ket{\Psi}$ optimized under all bipartitions \cite{Bourennane2004Experimental}. For instance, $\alpha=\frac1{2}$ for the GHZ state.
The expectation value of $\mathcal{W}$ shows, $\mathrm{Tr}(\mathcal{W}\rho_{pre})=\alpha-\mathrm{Tr}(\ket{\Psi}\bra{\Psi}\rho_{pre})$, which is directly related to measuring fidelity.

Normally, the multipartite projector $\ket{\Psi}\bra{\Psi}$ is decomposed with a few of local measurement settings (LMSs) \cite{Terhal2002Detecting,Guhne2002local}, for example the Pauli operator $\sigma_x^{\otimes N}$, which can be realized in experiments. Even for one LMS, it needs thousands of times of the measurement to obtain the estimation of the expectation value. Thus, the total number of LMSs characterize the efficiency of the witness. For the GHZ state, it needs $N+1$ LMSs \cite{Guhne2007Toolbox}.
On the other hand, the robustness is another key feature of a witness, which benchmarks its detection ability. Generally, one applies white noise tolerance to characterize the robustness, i.e.,
\begin{equation}\label{}
\begin{aligned}
\rho=(1-p)\ket{\Psi}\bra{\Psi}+p\frac{\mathbb{I}}{2^N},
\end{aligned}
\end{equation}
which moves the target state towards the maximal mixed state. The maximal $p_{max}$ such that $\mathrm{Tr}(\mathcal{W}\rho)\leq0$ describes the robustness of the witness. 

Since the witness $\mathcal{W}$ shown in Eq.~\eqref{Eq:witness} is designed specifically for the target state $\ket{\Psi}$, it may return null results for some other entangled states. This phenomenon may become serious when the experiment system suffers from the coherent noise, i.e.,
\begin{equation}\label{}
\begin{aligned}
\ket{\Psi_{\textrm{pre}}}=U_{\textrm{noise}}\ket{\Psi}.
\end{aligned}
\end{equation}
Since the unitary evolution can ''rotate" the state dramatically (not like the translation in the white noise case), the white noise tolerance corresponding to the result state can decrease, and is possibly outside the detection range of the witness in some case. See Fig.~\ref{Fig:CohNoise} for an illustration.
Taking the GHZ state as an example, according to Eq.~\eqref{Eq:witness}, the fidelity-based witness shows,
\begin{equation}\label{Eq:witnessGHZ}
\begin{aligned}
\mathcal{W}_{GHZ}=\frac{1}{2} \mathbb{I}-\ket{GHZ}\bra{GHZ},
\end{aligned}
\end{equation}
where $\ket{GHZ}=\frac1{\sqrt{2}}(\ket{0}^{\otimes N}+ \ket{1}^{\otimes N})$. If the prepared state becomes $\ket{\Psi_{pre}}=\frac1{\sqrt{2}}(\ket{0}^{\otimes N}-\ket{1}^{\otimes N})$ under some coherent error that affects the phase, the witness gives a null result $\mathrm{Tr}(\mathcal{W}_{GHZ}\ket{\Psi_{pre}}\bra{\Psi_{pre}})=\frac1{2}>0$. Note that here $\ket{\Psi_{pre}}$ is entangled but one cannot confirm this by using the witness $\mathcal{W}_{GHZ}$ in Eq.~\eqref{Eq:witnessGHZ}.

To the best of our knowledge, the entanglement detection under realistic coherent noises still lacks studying. The investigation along this direction can offer us two main advantages. On the one hand, it can supply useful tools to tackle with coherent noises and hence enhance our entanglement detection ability; on the other hand, it is also helpful to the benchmarking and even the calibration of experimental systems. This is beneficial for the ultimate goal---fault-tolerant quantum computation \cite{Nielsen2011Quantum,gottesman1997stabilizer}, as the coherent noise leads to a much worse threshold than the stochastic ones \cite{Aliferis2006Threshold}.

In this work, we study the entanglement detection under coherent noises and focus on the GHZ state, which is essential in many quantum information tasks, such as Bell-nonlocality \cite{Brunner2014nonlocality}, multipartite quantum key distribution \cite{Chen2007Multi}, quantum secret sharing \cite{Hillery1999secret,Cleve1999Share}, and quantum metrology \cite{Wineland1992squeezing,Giovannetti2004measure}. We show an entanglement detection protocol that can effectively eliminate the influence of certain types of coherent noises for the GHZ state. Our protocol only adds $1$ LMS than the original one, which needs $N+1$ LMSs. In particular, the protocol can effectively realize a family of entanglement witnesses with respective to the coherent noise, and one can select the optimal one by only postprocessing the measurement results. The protocol can also help us to estimate corresponding noise parameters and further give feedback to the experiment system. Moreover, we also consider the reduction of the number of LMSs, which makes the experimental realizations more efficient.

Our paper is organized as follows. In Sec.~\ref{Sec:noiseM}, two coherent noise models of the GHZ state are proposed, one is generated by the unconscious phase accumulation on $N$ qubits, the other is due to the rotation on the control qubit. The overall noise model is the combination of the coherent part and the white noise part. In Sec.~\ref{Sec:EWfull}, we show the detection protocol with $N+2$ LMSs, used to witness the entanglement under coherent noises and further feed back the noisy parameters.
In Sec.~\ref{Sec:EWeff}, we further reduce the number of LMSs to $3$, and propose more efficient witnesses.
Sec.~\ref{Sec:conclude} is the conclusion and outlook.

\begin{figure}[thb]
\centering
\resizebox{6cm}{!}{\includegraphics[scale=1]{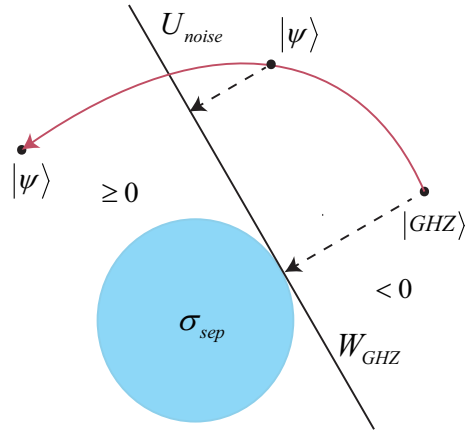}}
\caption{The effect of the coherent noise and the failure of the entanglement detection. The blue disk labels the convex separable state set $\mathcal{S}_{sep}$, and the witness $\mathcal{W}_{GHZ}$ is represented by the right solid line (a hyperplane in the state space) that is tangent to the disk. The white noise displaces the GHZ state towards the witness line, and the length of this dotted arrowed line can reveal the white noise tolerance. The red curve, labeling the coherent noise, ``rotates" the GHZ state to some $\ket{\Psi}$. Due to the coherent noise, the white noise tolerance with respective to $\mathcal{W}_{GHZ}$ decreases, as shown by the shortening of the dotted line. Finally, $\ket{\Psi}$ can moves to the other side of the hyperplane, thus its entanglement cannot be witnessed using $\mathcal{W}_{GHZ}$.}
\label{Fig:CohNoise}
\end{figure}
\section{The noise model}\label{Sec:noiseM}
In this section, we show two realistic coherent noise models of the GHZ state, which will be analysed in the following Sec \ref{Sec:EWfull}. One is caused by the unconscious phase accumulating on all qubits, the other is due to the single qubit rotation on the first control qubit.

Let us first review the white noise model.
Usually, one uses the white noise to analyse the noise tolerance of the entanglement witness, i.e., mixing the original state with the maximally mixed state,
\begin{equation}\label{}
\begin{aligned}
\Phi_p(\rho)=(1-p)\rho+p\frac{\mathbb{I}}{2^N}.
\end{aligned}
\end{equation}
For the GHZ state, the resulting state is
\begin{equation}\label{}
\begin{aligned}
\rho_w^p=(1-p)\ket{GHZ}\bra{GHZ}+p\frac{\mathbb{I}}{2^N}.
\end{aligned}
\end{equation}
The corresponding noise tolerance is determined by $\mathrm{Tr}(\mathcal{W}_{GHZ}\rho_w^p)=0$ where $\mathcal{W}_{GHZ}$ is defined in Eq.~\eqref{Eq:witnessGHZ}, and it equals to,
\begin{equation}\label{Eq:WhTol}
\begin{aligned}
p_{max}=\frac{2^{(N-1)}}{{2^N-1}}\ \ (\simeq 0.5,N \rightarrow \infty).
\end{aligned}
\end{equation}
The white noise is generated by the depolarizing channel, and it effectively displaces the original state $\ket{GHZ}$ towards the maximally mixed state in the state space, as shown in Fig. \ref{Fig:CohNoise}.  However, generally speaking, the coherent noise could appear in the experiment due to some system errors, as we illustrated in the following sections.

\subsection{Model 1: Unconscious phase accumulation}\label{Ssec:phase}
In experimental realizations, the degree of freedom of N-qubit is generally encoded in $N$ two-level subsystems, such as the ground state and the excited state of atoms. There might appear unconscious phase accumulation between $\ket{0}$ and $\ket{1}$ of qubits that dramatically transforms the state. To be specific, this kind of coherent error can be modeled as,
\begin{equation}\label{Eq:CnoiseM}
\begin{aligned}
&U_z=\bigotimes_{j=1}^N U_j^z,\\
&U_j^z=\ket{0}_j\bra{0}+e^{i\phi_j}\ket{1}_j\bra{1},
\end{aligned}
\end{equation}
where $U_j^z$ denote the rotation around the $Z$-basis on the $j$-th qubit.
If we apply the above coherent noise on the GHZ state, it shows
\begin{equation}\label{Eq:stateCN}
\begin{aligned}
\ket{\Psi_{\phi}}&=U_z\ket{GHZ}\\
&=\bigotimes_{j=1}^N U_j^z\ket{GHZ}\\
&=\frac1{\sqrt{2}}\left(\ket{0}^{\otimes N}+e^{i\sum_{j=1}^N\phi_j} \ket{1}^{\otimes N}\right)\\
&=\frac1{\sqrt{2}}(\ket{0}^{\otimes N}+e^{i\phi} \ket{1}^{\otimes N}),\\
\end{aligned}
\end{equation}
where $\phi\doteq\sum_{j=1}^N\phi_j$. Similar as the white noise case, the tolerance of  $\mathcal{W}_{GHZ}$ in Eq.~\eqref{Eq:witnessGHZ} under this coherent noise is determined by
\begin{equation}\label{Eq:noiseTcoh}
\begin{aligned}
\mathrm{Tr}(\mathcal{W}_{GHZ}\ket{\Psi_{\phi}}\bra{\Psi_{\phi}})=-\frac{\cos \phi} {2}<0
\end{aligned}
\end{equation}
which leads to $\phi\in (-\frac{\pi}{2}, \frac{\pi}{2})$. Thus if the absolute value of the phase $|\phi|\geq \frac{\pi}{2}$, the witness $\mathcal{W}_{GHZ}$ cannot properly detect the entanglement, while the prepared state $\ket{\Psi_{\phi}}$ is clearly an entangled one.

More generally, the realistic noise can be the combination of the white noise part $\Phi_p$ and the coherent part $U_z$, thus the output state shows,
\begin{equation}\label{Eq:Snoisy}
\begin{aligned}
\rho_{pre}&=\Phi_p \circ U_z(\ket{GHZ}\bra{GHZ})\\
&=(1-p)\ket{\Psi_{\phi}}\bra{\Psi_{\phi}}+p\frac{\mathbb{I}}{2^N},
\end{aligned}
\end{equation}
where $U_z(\cdot)=U_z\cdot U_z^{\dag}$, and note that $\Phi_p \circ U_z= U_z \circ \Phi_p$.

In this joint noise model, the noise tolerance range is determined by $\mathrm{Tr}(\mathcal{W}_{GHZ}\rho_{pre})<0$ with $\rho_{pre}$ in Eq.~\eqref{Eq:Snoisy}. The result is given by the following formula including the coherent and white noise parameters $\phi$ and $p$,
\begin{equation}\label{Eq:noiseTcom}
\begin{aligned}
\cos \phi >\frac{p}{1-p},
\end{aligned}
\end{equation}
as $N \rightarrow \infty$. The detailed derivation is left in Appendix \ref{Sec:noiseT}.
Comparing to Eq.~\eqref{Eq:noiseTcoh}, Eq.~\eqref{Eq:noiseTcom} shows that the range of $\phi$ shrinks due to the introduction of white noise. On the other hand, Eq.~\eqref{Eq:noiseTcom} can be rewritten as follows,
\begin{equation}\label{Eq:noiseTcomR}
\begin{aligned}
p<1-\frac{1}{\cos\phi+1}.
\end{aligned}
\end{equation}
It indicates that the range of the white noise parameter $p$ also decreases on account of the coherent noise, comparing to Eq.~\eqref{Eq:WhTol}.

\subsection{Model 2: Rotation on the first control qubit}\label{Ssec:rotation}
The GHZ state is normarlly generated by the following circuit routine, as shown in Fig.~\ref{Fig:GHZcir}.
\begin{itemize}
  \item Initialize all the $N$ qubits to be $\ket{0}$.
  \item Apply a Hardmard gate $H$ on the first (control) qubit, and transform it to $\ket{+}=\frac1{2}(\ket{0}+\ket{1})$.
  \item Apply Controled-NOT (CNOT) gate on qubit pairs $(1,2), (2,3), (3,4),\cdots(j,j+1)\cdots\cdots(N-1,N)$ in sequence, where $j$ is the control qubit and $j+1$ is the target qubit.
\end{itemize}

It is clear to see that the CNOT gate sequence spreads the superposition information of the first qubit to all the qubits, and thus builds the quantum correlation on the whole system.
Hence, the quality of the rotation on the first qubit significantly affects the preparation of the final GHZ state.
\begin{figure}[thb]
\centering
\resizebox{7cm}{!}{\includegraphics[scale=1]{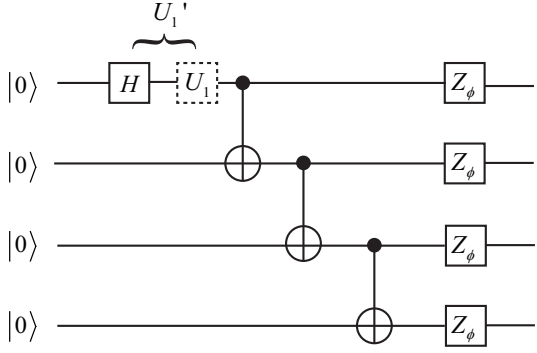}}
\caption{The quantum circuit to generate the 4-qubit GHZ state and the coherent noise on the 1st control qubit. The overall noisy unitary is denoted by $U_1'=U_1H_1$. Note that at the end of the circuit, we also allow the Z-basis phase accumulation $Z_{\phi}$ where $\phi$ needs not to be the same for each qubit.}
\label{Fig:GHZcir}
\end{figure}

Suppose besides the ideal $H$ gate, there is also another uncontrolled unitary on the first qubit, i.e.,
\begin{equation}\label{Eq:singleU}
\begin{aligned}
\ket{\psi}_1=U_1H_1\ket{0}_1=\cos\theta\ket{0}_1 + e^{i\phi}\sin\theta\ket{1}_1,
\end{aligned}
\end{equation}
with $\theta\in[0,\frac{\pi}{2}]$ and $\phi\in[-\pi,\pi)$.
Here the overall unitary $U_1'=U_1H_1$ in principle can be any single qubit unitary, thus $\ket{\psi}_1$ describes any single qubit state after ignoring the irrelevant global phase. In addition, we also allow the
unconscious phase accumulation on the state at the final stage.

Consequently, the final prepared state shows,
\begin{equation}\label{Eq:NoiseR}
\begin{aligned}
\ket{\Psi_{\phi}^{\theta}}=\cos\theta\ket{0}^{\otimes N}+e^{i\phi}\sin\theta \ket{1}^{\otimes N},
\end{aligned}
\end{equation}
where the accumulated phase at the final stage is also dropped into the parameter $\phi$ without confusion. Note that the noisy state $\ket{\Psi_{\phi}^{\theta}}$ in Eq.~\eqref{Eq:NoiseR} is more general than $\ket{\Psi_{\phi}}$ in Eq.~\eqref{Eq:stateCN}, since $\ket{\Psi_{\phi}^{\theta}}$ allows unbalanced state coefficients besides the relative phase.

The noise tolerance of $\mathcal{W}_{GHZ}$ in Eq.~\eqref{Eq:witnessGHZ} under this coherent noise is determined by
\begin{equation}\label{Eq:noiseTcoh1}
\begin{aligned}
\mathrm{Tr}(\mathcal{W}_{GHZ}\ket{\Psi_{\phi}^{\theta}}\bra{\Psi_{\phi}^{\theta}})=\frac1{2}-\frac1{2}[1+\sin(2\theta)\cos(\phi)]<0
\end{aligned}
\end{equation}
that is, $\sin(2\theta)\cos(\phi)>0$, which leads to $\phi\in (-\frac{\pi}{2}, \frac{\pi}{2})$.

As in Sec. \ref{Ssec:phase}, one can also consider the combination of the coherent noise and the white noise, and the final state shows,
\begin{equation}\label{Eq:Snoisy2}
\begin{aligned}
\rho_{pre}=(1-p)\ket{\Psi_{\phi}^{\theta}}\bra{\Psi_{\phi}^{\theta}}+p\frac{\mathbb{I}}{2^N},
\end{aligned}
\end{equation}

Accordingly, the tolerance of $\mathcal{W}_{GHZ}$ in this scenario when $N \rightarrow \infty$ shows,
\begin{equation}\label{Eq:noiseTcom2}
\begin{aligned}
\sin(2\theta)\cos(\phi)> \frac{p}{1-p}.
\end{aligned}
\end{equation}
The detailed derivation is left in Appendix \ref{Sec:noiseT}. Comparing to Eq.~\eqref{Eq:noiseTcom}, one can see that the noise tolerance range decreases further, after the introduction of the noise parameter $\theta$. In addition, Eq.~\eqref{Eq:noiseTcom2} can be rewritten as,
\begin{equation}\label{Eq:noiseTcomR2}
\begin{aligned}
p<1-\frac{1}{\sin(2\theta)\cos(\phi)+1},
\end{aligned}
\end{equation}
and it is worse than Eq.~\eqref{Eq:WhTol} and \eqref{Eq:noiseTcomR}. See Fig.~\ref{Fig:WNGHZ} for an illustration.
\begin{figure}[thb]
\centering
\resizebox{9.5cm}{!}{\includegraphics[scale=1]{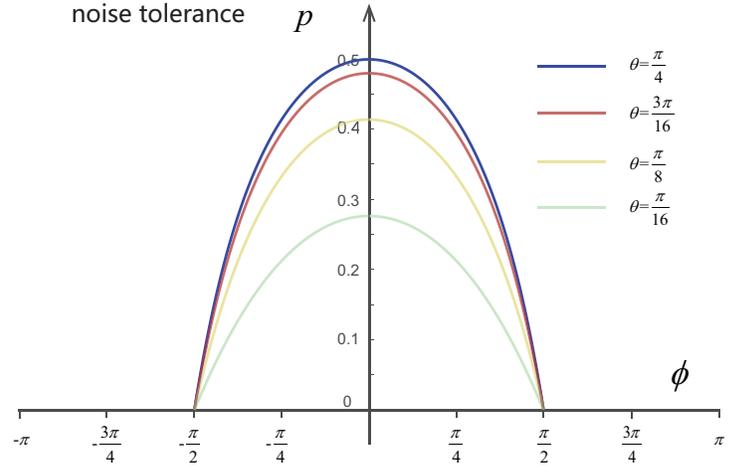}}
\caption{The decrease of the noise tolerance range after introducing coherent noises, as shown in Eq.~\eqref{Eq:noiseTcomR2}. We plot the white noise parameter $p$ as a function of $\phi$ for different $\theta$. The area under the curve labels the parameter region where the corresponding state can be detected by the witness $\mathcal{W}_{GHZ}$ in Eq.~\eqref{Eq:witnessGHZ}. For the top (blue) curve with $\theta=\frac{\pi}{4}$, one can see that $p$ decreases for larger coherent noisy parameter $\phi$. After further introducing the parameter $\theta$, the noise tolerance decreases further as $\theta$ departs from $\frac{\pi}{4}$.}
\label{Fig:WNGHZ}
\end{figure}
\section{Entanglement detection protocol under coherent noise}\label{Sec:EWfull}
As shown in Sec.~\ref{Sec:noiseM}, the witness $\mathcal{W}_{GHZ}$ specialized for the GHZ state potentially returns a null result when the prepared state suffers from some coherent noise. Here we propose an entanglement detection protocol that can eliminate the effect of coherent noises shown in the above section. The protocol only involves $N+2$ LMSs, which only adds one LMS comparing to the previous witness $\mathcal{W}_{GHZ}$ specialized for the GHZ state \cite{Guhne2007Toolbox}.

Since the resulting state of noise model 2 in Sec \ref{Ssec:rotation} is more general than that of model 1 in Sec \ref{Ssec:phase}, for clearness, in the following we first apply the entanglement detection protocol on the model 1, and then generalize it to the model 2.

\begin{figure}[thb]
\centering
\resizebox{6.5cm}{!}{\includegraphics[scale=1]{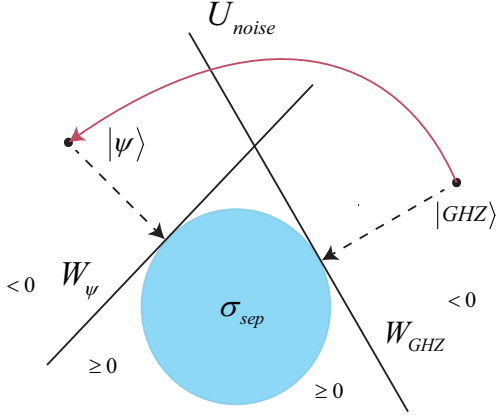}}
\caption{The illustration of the entanglement detection protocol. The blue disk labels the separable state set $\mathcal{S}_{sep}$, and the witness $\mathcal{W}_{GHZ}$ is represented by the right solid line. The length of this dotted arrowed line can reveal the white noise tolerance. The red curve, labeling the coherent noise, ``rotates" the GHZ state to some $\ket{\Psi}$. Here we effectively realize a family of witnesses $\mathcal{W}_{\Psi}$ as shown in Eq.~\eqref{Eq:witnessF} and \eqref{Eq:witnessF2}, and select the \emph{optimal} one by postprocessing the measurement results. Comparing to the situation in Fig.~\ref{Fig:CohNoise}, the protocol makes the entanglement detection possible again.}
\label{Fig:CohNoise2}
\end{figure}
\subsection{Detection protocol under noise model 1}\label{sSec:EWnm1}
The protocol measures the fidelity between $\rho_{pre}$ and $\ket{\Psi_{\phi}}$ in Eq.~\eqref{Eq:stateCN} for any phase parameter $\phi$ with the \emph{same} $N+2$ LMSs. As a result, one can effectively detect the  entanglement by choosing the \emph{optimal} witness in the family,
\begin{equation}\label{Eq:witnessF}
\begin{aligned}
\mathcal{W}_{\Psi_{\phi}}=\frac1{2} \mathbb{I}-\ket{\Psi_{\phi}}\bra{\Psi_{\phi}},
\end{aligned}
\end{equation}
by postprocessing the measurement results. See Fig.~\ref{Fig:CohNoise2} for an illustration. Hereafter qubit Pauli operators are denoted by $\{\sigma_x, \sigma_y, \sigma_z\}$, and we summarize the result into the following Theorem.
\begin{theorem}\label{Th:fullLM1}
The family of witnesses $\mathcal{W}_{\Psi_{\phi}}$ parameterized by $\phi$ in Eq.~\eqref{Eq:witnessF} can be realized with totally $N+2$ LMSs, i.e., $\sigma_z^{\otimes N}$ and
\begin{equation}\label{Eq:LMSN}
\begin{aligned}
 \mathcal{M}_{\theta_k}=\left(\cos \theta_k\sigma_x+ \sin \theta_k\sigma_y\right)^{\otimes N},
\end{aligned}
\end{equation}
where $\theta_k=\frac{k\pi}{N+1}$ and $k=0,1,\cdots,N$.
\end{theorem}

\begin{proof}
The projector $\ket{\Psi_{\phi}}\bra{\Psi_{\phi}}$ can be written as,
\begin{equation}\label{Eq:DnoiseS}
\begin{aligned}
\ket{\Psi_{\phi}}\bra{\Psi_{\phi}}=\mathcal{Z}+\mathcal{X},
\end{aligned}
\end{equation}
where $\mathcal{Z}$ denotes the summation of diagonal terms, i.e.,
\begin{equation}\label{Eq:noiseD}
\begin{aligned}
\mathcal{Z}=&\frac1{2}(\ket{0}\bra{0}^{\otimes N}+\ket{1}\bra{1}^{\otimes N}),\\
\end{aligned}
\end{equation}
and $\mathcal{X}$ is for off-diagonal terms
\begin{equation}\label{Eq:noiseOff}
\begin{aligned}
\mathcal{X}=\cos\phi\mathcal{X}_+ +\sin\phi\mathcal{X}_-,\\
\end{aligned}
\end{equation}
where
\begin{equation}\label{Eq:noiseOff1}
\begin{aligned}
&\mathcal{X}_+=\frac{\ket{0}\bra{1}^{\otimes N}+\ket{1}\bra{0}^{\otimes N}}{2},\\
&\mathcal{X}_-=\frac{\ket{0}\bra{1}^{\otimes N}-\ket{1}\bra{0}^{\otimes N}}{2i}.\\
\end{aligned}
\end{equation}

The diagonal part $\mathcal{Z}$ can be measured with the LMS $\sigma_z^{\otimes N}$. The off-diagonal part
$\mathcal{X}_+$ and $\mathcal{X}_-$ involved in $\mathcal{X}$ can be further decomposed with LMSs $\mathcal{M}_{\theta_k}$ given in Eq.~\eqref{Eq:LMSN} as,
\begin{equation}\label{Eq:DecOff}
\begin{aligned}
&\mathcal{X}_+=\frac1{N+1}\sum_{k=0}^N (-1)^k \cos(\theta_k) \mathcal{M}_{\theta_k}\\
&\mathcal{X}_-=\frac{-1}{N+1}\sum_{k=0}^N (-1)^k \sin(\theta_k) \mathcal{M}_{\theta_k}.
\end{aligned}
\end{equation}
The proof of these decompositions is based on discrete Fourier transform, and we leave it in Appendix \ref{prf:thm1}.
\end{proof}

To eliminate the effect of the coherent noise due to the unconscious phase accumulation, one should maximize the fidelity between the prepared state $\rho_{pre}$ and all possible $\ket{\Psi_{\phi}}$ based on measurement results, that is,
\begin{equation}\label{Eq:Nmax}
\begin{aligned}
\max_{\phi}\bra{\Psi_{\phi}}\rho_{pre}\ket{\Psi_{\phi}}&=\langle \mathcal{Z}\rangle+\max_{\phi}\langle \mathcal{X}\rangle\\
&=\langle \mathcal{Z}\rangle + \max_{\phi}\left\{ \cos\phi \langle \mathcal{X}_+\rangle + \sin\phi \langle \mathcal{X}_-\rangle\right\}\\
&=\langle \mathcal{Z}\rangle+\sqrt {\langle \mathcal{X}_+\rangle ^2+ \langle \mathcal{X}_-\rangle^2},
\end{aligned}
\end{equation}
where $\langle\cdot\rangle$ denotes the expectation value of the corresponding operator on $\rho_{pre}$, and in the final line we apply the Cauchy-Schwarz inequality. Note that $\langle \mathcal{Z}\rangle$ and $\langle \mathcal{X}_{\pm}\rangle$ can be obtained from LMS $\sigma_z$ and $\{\mathcal{M}_{\theta_k}\}_{k=0}^N$, respectively. The optimal $\phi_{opt}$ to saturate the maximal value in Eq.~\eqref{Eq:Nmax} is determined by
\begin{equation}\label{Eq:OpPhi}
\begin{aligned}
\tan \phi_{opt}&=\frac{\langle \mathcal{X}_-\rangle}{\langle \mathcal{X}_+\rangle}\\
&=-\frac{\sum_{k=0}^N (-1)^k \sin(\theta_k) \langle\mathcal{M}_{\theta_k}\rangle}{\sum_{k=0}^N (-1)^k \cos(\theta_k) \langle\mathcal{M}_{\theta_k}\rangle}.
\end{aligned}
\end{equation}
where the second line is on account of Eq.~\eqref{Eq:DecOff}, and $(\cos\phi_{opt}, \sin\phi_{opt})$ is in the same quadrant with $(\langle \mathcal{X}_+\rangle, \langle \mathcal{X}_-\rangle)$.

For instance, for the noisy state shown in Eq.~\eqref{Eq:Snoisy}, one can effectively choose the corresponding witness in Eq.~\eqref{Eq:witnessF} to eliminate the effect of the coherent noise and detect the entanglement. Note that the parameter $\phi_{opt}$ is determined by the measurement results. It is clear that the noise tolerance now is the same as in the sole white noise case, $p_{max}$ in Eq.~\eqref{Eq:WhTol}, no matter what value $\phi$ is.

Moreover, this protocol can further help to improve the experiment system. That is, one can apply an reverse unitary to amend the system according to the optimal $\phi_{opt}$ abstracted from the measurement results. In particular, one can add a corresponding Z-basis rotation on any qubit to eliminate the error.

\subsection{Detection protocol under noise model 2}\label{sSec:EWnm2}
In this section, we generalize the entanglement detection protocol proposed in Sec.~\ref{sSec:EWnm1} and apply it to the noise model 2.

The main strategy is similar, and here we realize the following family of witnesses with the same $N+2$ LMSs.
\begin{equation}\label{Eq:witnessF2}
\begin{aligned}
\mathcal{W}_{\Psi_{\phi}^{\theta}}=\max\{\cos^2\theta, \sin^2\theta\} \mathbb{I}-\ket{\Psi_{\phi}^{\theta}}\bra{\Psi_{\phi}^{\theta}},
\end{aligned}
\end{equation}
where $\max\{\cos^2\theta, \sin^2\theta\}$ is the maximal Schmidt coefficient of $\ket{\Psi_{\phi}^{\theta}}$ defined in Eq.~\eqref{Eq:NoiseR}. One can further choose the \emph{optimal} witness in the family by post-processing the measurement results. We summarize this into the following Theorem.
\begin{theorem}\label{Th:fullLM2}
The family of witnesses $\mathcal{W}_{\Psi_{\phi}^{\theta}}$ parameterized by $\phi$ and $\theta$ in Eq.~\eqref{Eq:witnessF2} can be realized with totally $N+2$ LMSs, i.e., $\sigma_z^{\otimes N}$ and $\{\mathcal{M}_{\theta_k}\}_{k=0}^N$ defined in Eq.~\eqref{Eq:LMSN}.
\end{theorem}
\begin{proof}
As in Eq.~\eqref{Eq:DnoiseS}, the projector $\ket{\Psi_{\phi}^{\theta}}\bra{\Psi_{\phi}^{\theta}}$ can be decomposed as follows,
\begin{equation}\label{Eq:DnoiseS1}
\begin{aligned}
\ket{\Psi_{\phi}^{\theta}}\bra{\Psi_{\phi}^{\theta}}=\cos^2\theta\mathcal{Z}_0 + \sin^2\theta\mathcal{Z}_1 +\sin (2\theta) \mathcal{X},
\end{aligned}
\end{equation}
where $\mathcal{Z}_0$ and $\mathcal{Z}_1$ denote $\ket{0}\bra{0}^{\otimes N}$ and $\ket{1}\bra{1}^{\otimes N}$, whose expectation values can be evaluated from the LMS $\sigma_z^{\otimes N}$; $\mathcal{X}$ is given by Eq.~\eqref{Eq:noiseOff} and \eqref{Eq:noiseOff1}, whose expectation value can be obtained from LMSs $\{\mathcal{M}_{\theta_k}\}_{k=0}^N$, as shown in Eq.~\eqref{Eq:DecOff}.
\end{proof}


Similar as Sec.~\ref{sSec:EWnm1}, we should find the maximal fidelity between the prepared state and all possible $\ket{\Psi_{\phi}^{\theta}}$ based on the measurement results,
\begin{equation}\label{Eq:Nmax2}
\begin{aligned}
&\max_{\phi,\theta}\bra{\Psi_{\phi}^{\theta}}\rho_{pre}\ket{\Psi_{\phi}^{\theta}}=\max_{\phi,\theta}\Big{\{}\cos^2\theta \langle\mathcal{Z}_{0} \rangle+ \sin^2\theta \langle\mathcal{Z}_{1}\rangle + \\
&\ \ \ \ \ \ \ \ \ \ \sin (2\theta) ( \cos\phi \langle \mathcal{X}_+\rangle + \sin\phi \langle \mathcal{X}_-\rangle)\Big{\}}\\
&=\max_{\theta}\left\{\cos^2\theta \langle\mathcal{Z}_{0} \rangle+ \sin^2\theta \langle\mathcal{Z}_{1}\rangle + \sin (2\theta) \sqrt {\langle \mathcal{X}_+\rangle ^2+ \langle \mathcal{X}_-\rangle^2}\right\}\\
&=\max_{\theta}\Big{\{}\frac{\langle\mathcal{Z}_{0}\rangle+\langle\mathcal{Z}_{1}\rangle}{2}+ \cos(2\theta) \frac{\langle\mathcal{Z}_{0}\rangle-\langle\mathcal{Z}_{1}\rangle}{2} + \\
&\ \ \ \ \ \ \ \ \ \ \sin (2\theta) \sqrt {\langle \mathcal{X}_+\rangle ^2+ \langle \mathcal{X}_-\rangle^2}\Big{\}}\\
&=\frac{\langle\mathcal{Z}_{0}\rangle+\langle\mathcal{Z}_{1}\rangle}{2}+\sqrt{\frac1{4}(\langle\mathcal{Z}_{0}\rangle-\langle\mathcal{Z}_{1}\rangle)^2+\langle\mathcal{X}_+\rangle ^2+ \langle \mathcal{X}_-\rangle^2}.
\end{aligned}
\end{equation}
Here the maximization on the parameters $\phi$ and $\theta$ can be conducted independently. In the second line, we take the optimal $\phi_{opt}$ given by Eq.~\eqref{Eq:OpPhi}. The last line is due to the Cauchy-Schwarz inequality, and the optimal $\theta_{opt}$ takes the value,
\begin{equation}\label{Eq:OpTheta}
\begin{aligned}
\tan (2\theta_{opt})=\frac{2\sqrt {\langle \mathcal{X}_+\rangle ^2+ \langle \mathcal{X}_-\rangle^2}}{\langle\mathcal{Z}_{0}\rangle-\langle\mathcal{Z}_{1}\rangle},
\end{aligned}
\end{equation}
with $[\cos(2\theta_{opt}), \sin(2\theta_{opt})]$ being in the same quadrant with $\left[\langle\mathcal{Z}_{0}\rangle-\langle\mathcal{Z}_{1}\rangle, 2\sqrt {\langle \mathcal{X}_+\rangle ^2+ \langle \mathcal{X}_-\rangle^2}\right]$.
Then one can choose the optimal witness in the family of Eq.~\eqref{Eq:witnessF2} to detect the entanglement, based on the fidelity maximization in Eq.~\eqref{Eq:Nmax2} and the associated optimal parameters $\phi_{opt}$ and $\theta_{opt}$.

For instance, for the noisy state in Eq.~\eqref{Eq:Snoisy2}, it is not hard to see that the noise tolerance shows
\begin{equation}\label{Eq:Tol2}
\begin{aligned}
p<\min\{\cos^2\theta_{opt}, \sin^2\theta_{opt}\},
\end{aligned}
\end{equation}
with $N \rightarrow \infty$, no matter what value $\phi$ is. The detailed derivation is left in Appendix \ref{Sec:noiseT}. Note that the white noise tolerance is still a function of $\theta$, even if one can obtain its value by postprocessing. The reason is because the parameter $\theta$, not like $\phi$, indeed affects the entanglement.

On the other hand, one can also choose the optimal witness in the family of Eq.~\eqref{Eq:witnessF} in Sec.~\ref{sSec:EWnm1} on the noise model 2 here. Since the optimization in Eq.~\eqref{Eq:Nmax} can help to determine the corresponding noise parameter $\phi_{opt}$, the optimal witness shows,
\begin{equation}\label{}
\begin{aligned}
\mathcal{W}_{\Psi_{\phi_{opt}}}=\frac1{2} \mathbb{I}-\ket{\Psi_{\phi_{opt}}}\bra{\Psi_{\phi_{opt}}}.
\end{aligned}
\end{equation}
As a result, for the noisy state in Eq.~\eqref{Eq:Snoisy2}, the corresponding white noise tolerance reads (see Appendix \ref{Sec:noiseT} for the derivation),
\begin{equation}\label{Eq:NT1on2}
\begin{aligned}
p<1-\frac{1}{\sin(2\theta)+1}.
\end{aligned}
\end{equation}
which shows a clear advantage comparing to Eq.~\eqref{Eq:noiseTcomR2} with the original witness $\mathcal{W}_{GHZ}$. Note that the term $\cos(\phi)$ is eliminated due to the postprocessing. Surprisingly, the white noise tolerance in Eq.~\eqref{Eq:NT1on2} is better than the one in Eq.~\eqref{Eq:Tol2}, as illustrated in Fig.~\ref{Fig:WNCompare}. We give a detailed comparison in Appendix \ref{Sec:CompareN}. The reason for this phenomenon may be as follows.  By using the family of witnesses in Eq.~\eqref{Eq:witnessF2}, one maximizes the fidelity between $\ket{\Psi_{\phi}^{\theta}}$ and the prepared state. However, the corresponding fidelity bound in the witness for the separable state, i.e., $\max\{\cos^2\theta, \sin^2\theta\}$, becomes larger and harder to violate.

\begin{figure}[thb]
\centering
\resizebox{7cm}{!}{\includegraphics[scale=1]{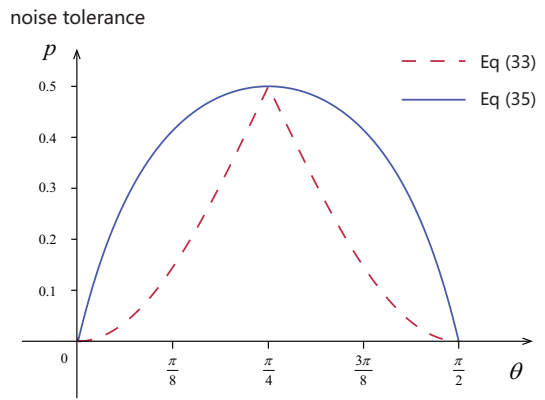}}
\caption{Comparison between the white noise tolerances in Eq.~\eqref{Eq:Tol2} and \eqref{Eq:NT1on2}. They are tolerances of the optimal witness in Eq.~\eqref{Eq:witnessF2} and Eq.~\eqref{Eq:witnessF} for the noisy state in Eq.~\eqref{Eq:Snoisy2}, respectively.  Both of them increase and reach the maximal value 0.5 at $\theta=\pi/4$, then decrease, since at this point the state coefficients are balanced and the state possesses the maximal entanglement.  Except $\theta=0,\pi/4, \pi/2$, Eq.~\eqref{Eq:NT1on2} always shows a clear advantage to Eq.~\eqref{Eq:Tol2}.}
\label{Fig:WNCompare}
\end{figure}


The entanglement detection protocol under the noise model 2 employs the same set of $N+2$ LMSs as that in Sec.~\ref{sSec:EWnm1}, but abstracts both noise parameters $\phi_{opt}$ and $\theta_{opt}$. This is because here we postprocess measurement results more delicately. Even though the white noise tolerance of the corresponding witness in Eq.~\eqref{Eq:witnessF2} is not better than the one in Eq.~\eqref{Eq:witnessF}, the experiment system can be further improved with the noise parameters $\phi_{opt}$ and $\theta_{opt}$ extracted from the measurement results. In particular, one can add an unitary $U_1^{\dag}$ on the first qubit when preparing the GHZ state, which can be determined by $\phi_{opt}$ and $\theta_{opt}$.

\section{Entanglement detection with less LMS}\label{Sec:EWeff}
In entanglement detection, the number of LMSs usually determines the efficiency of the witness, since even for one setting it could take thousands of measurements to obtain the accurate estimation of the expectation value. Thus, it is beneficial to reduce the number of LMSs and enhance the efficiency of the witness. In this section, by utilizing the stabilizer formulation, we show that one can detect entanglement of GHZ-like states with only $3$ LMSs under realistic coherent noises.

Comparing to the witness $\mathcal{W}_{GHZ}$ using $N+1$ LMSs in Eq.~\eqref{Eq:witnessGHZ},
there is a more efficient witness using $2$ LMSs \cite{Toth2005Detecting}, i.e.,
\begin{equation}\label{Eq:EWGHZtwo}
\begin{aligned}
\mathcal{W}^2_{GHZ}=\frac{1}{2} \mathbb{I}-\mathcal{Z}-\frac{1}{4}\sigma_x^{\otimes N},
\end{aligned}
\end{equation}
where $\mathcal{Z}$ is defined in Eq.~\eqref{Eq:noiseD}. However, there is a trade-off between the efficiency and the white noise tolerance \cite{Guhne2007Toolbox,Qi2019Efficient}. For the target GHZ state, the white noise tolerance of $\mathcal{W}^2_{GHZ}$ is $p=\frac{1}{3}$ as $N \rightarrow \infty$ \cite{Toth2005Detecting}, while the tolerance is $p=\frac{1}{2}$ of $\mathcal{W}_{GHZ}$. Note that $\mathcal{W}^2_{GHZ}$ employs two settings, i.e., $\sigma_z^{\otimes N}$ and $\sigma_x^{\otimes N}$.

In the following, we study the entanglement detection under coherent noises with less LMSs, by adding another LMS $\sigma_y\otimes \sigma_x ^{\otimes N-1}$ to $\sigma_z^{\otimes N}$ and $\sigma_x^{\otimes N}$. Actually, the $\sigma_y$ operator can be on any qubit due to the symmetry of GHZ-like states. With loss of generality, we set it on the first qubit. Similar as Sec.~\ref{Sec:EWfull}, we first consider the entanglement detection under noise model 1, then noise model 2 that contains more possible noisy states.
\subsection{Efficient detection under noise model 1}\label{sSec:EEWnm1}
We extend the witness $\mathcal{W}^2_{GHZ}$ in Eq.~\eqref{Eq:EWGHZtwo} to a family of witnesses, parameterized by the noisy parameter $\phi$. Similar as Sec.~\ref{sSec:EWnm1}, one can effectively detect entanglement under the coherent noise by choosing the \emph{optimal} witness in this family, by postprocessing the measurement results. We summarize the result in the following theorem.
\begin{theorem}\label{Th:2LM1}
The witness $\mathcal{W}_{\Psi_{\phi}}^2$ can detect entanglement near the state $\ket{\Psi_{\phi}}$,
\begin{equation}\label{Eq:witnessTwo}
\begin{aligned}
\mathcal{W}_{\Psi_{\phi}}^2=\frac{1}{2} \mathbb{I}-\mathcal{Z}-\frac{1}{4}(\cos \phi \mathcal{M}_x + \sin \phi \mathcal{M}'_x),
\end{aligned}
\end{equation}
where $\mathcal{Z}$ is defined in Eq.~\eqref{Eq:noiseD}, $\mathcal{M}_x=\sigma_x ^{\otimes N}$, and $\mathcal{M}_x'=\sigma_y\otimes \sigma_x ^{\otimes N-1}$.
\end{theorem}
It is clear that the family of witness $\{\mathcal{W}_{\Psi_{\phi}}^2\}$ parameterized by $\phi$ can be realized by $3$ LMSs, $\sigma_z ^{\otimes N},\sigma_x ^{\otimes N},\sigma_y\otimes \sigma_x ^{\otimes N-1}$.
\begin{proof}
As given in Eq.~\eqref{Eq:stateCN}, the possible state under the coherent noise shows $\ket{\Psi_{\phi}}=\frac1{\sqrt{2}}(\ket{0}^{\otimes N}+e^{i\phi} \ket{1}^{\otimes N})$, and it can be transformed from the standard GHZ state by applying a single qubit unitary,
\begin{equation}\label{}
\begin{aligned}
\ket{\Psi_{\phi}}=U_1^z\otimes \mathbb{I}^{\otimes N-1} \ket{GHZ},
\end{aligned}
\end{equation}
where $U_1^z=\ket{0}_1\bra{0}+e^{i\phi}\ket{1}_1\bra{1}=e^{i\phi/2}e^{-i\sigma_z\phi/2}$. For the density matrix, one has the relation $\Psi_{\phi}=U_1^z GHZ U_1^{z\dag}$. Since the witness $\mathcal{W}^2_{GHZ}$ can detect entanglement near the GHZ state, we have the following witness that can detect entanglement near $\ket{\Psi_{\phi}}$ based on Observation \ref{Ob:LU},
\begin{equation}\label{}
\begin{aligned}
\mathcal{W}_2^{\phi}&=U_1^z \mathcal{W}^2_{GHZ} U_1^{z\dag}\\
&=\frac{1}{2}\mathbb{I}-\mathcal{Z}-\frac{1}{4}(e^{-i\sigma_z\phi/2}\sigma_xe^{i\sigma_z\phi/2})\otimes\sigma_x^{\otimes N-1},\\
&=\frac{1}{2}\mathbb{I}-\mathcal{Z}-\frac{1}{4}(\cos \phi \sigma_x + \sin \phi \sigma_y)\otimes\sigma_x^{\otimes N-1}.
\end{aligned}
\end{equation}
In this way, we obtain a family of entanglement witnesses parameterized by the phase $\phi$, and they all need $2$ LMSs, i.e. $\sigma_z^{\otimes N}$ and $(\cos \phi \sigma_x - \sin \phi \sigma_y)\otimes\sigma_x^{\otimes N-1}$. In fact, one can realize these witnesses with only $3$ LMSs, $\mathcal{M}_z=\sigma_z^{\otimes N}$, $\mathcal{M}_x=\sigma_x^{\otimes N}$and $\mathcal{M}_x'=\sigma_y\otimes\sigma_x^{\otimes N-1}$, since the result of $(\cos \phi \sigma_x - \sin \phi \sigma_y)\otimes\sigma_x^{\otimes N-1}$ can be obtained by the linear combination of $\mathcal{M}_x$ and $\mathcal{M}_x'$.
\end{proof}

\begin{observation}\label{Ob:LU}
Suppose an entanglement witness $\mathcal{W}$ can detect a entangled state $\rho$, i.e., $\mathrm{Tr}(\mathcal{W}\rho)<0$ and $\mathrm{Tr}(\mathcal{W}\sigma_s)\geq0,\ \forall\sigma_s \in S_{sep}$, the state $\rho'=U_{loc}\rho U_{loc}^{\dag}$ after the transformation with local unitary operation $U_{loc}$ can be detected by the corresponding witness
\begin{equation}\label{}
\begin{aligned}
\mathcal{W}'=U_{loc}\mathcal{W}U_{loc}^{\dag}
\end{aligned}
\end{equation}
where $U_{loc}=U_1\otimes U_2\cdots U_N$ and $U_i$ is the unitary on the $i$-th qubit.
\end{observation}
\begin{proof}
Note that $\rho'$ is still an entangled state, as local unitary operations does not alter entanglement property. $\mathrm{Tr}(\mathcal{W}'\rho')=\mathrm{Tr}(U_{loc}\mathcal{W}U_{loc}^{\dag}U_{loc}\rho U_{loc}^{\dag})=\mathrm{Tr}(\mathcal{W}\rho)<0$. In addition, $\mathrm{Tr}(\mathcal{W}'\sigma_{sep})=\mathrm{Tr}(U_{loc}\mathcal{W}U_{loc}^{\dag}\sigma_{sep})=\mathrm{Tr}(\mathcal{W}U_{loc}^{\dag}\sigma_{sep}U_{loc})\geq0$,  since $U_{loc}^{\dag}\sigma_{sep}U_{loc}$ is still a separable state.
\end{proof}
To eliminate the effect of the coherent noise, similar as Eq.~\eqref{Eq:Nmax}, one should find the minimal expectation value of all the witnesses in the family, i.e., $\min_{\phi} \mathrm{Tr}(\mathcal{W}_{\Psi_{\phi}}^2\rho_{pre})$. Equivalently,
\begin{equation}\label{Eq:2min}
\begin{aligned}
\max_{\phi} \{\cos \phi \langle\mathcal{M}_x\rangle+\sin \phi\langle \mathcal{M}_x' \rangle\}.
\end{aligned}
\end{equation}
The optimal $\phi$ is determined by
\begin{equation}\label{Eq:optPhiTwo1}
\begin{aligned}
\tan \phi_{opt}=\frac{\langle\mathcal{M}_x'\rangle}{\langle\mathcal{M}_x\rangle}.
\end{aligned}
\end{equation}
and $(\cos\phi_{opt}, \sin\phi_{opt})$ is in the same quadrant with $(\langle\mathcal{M}_x\rangle, \langle\mathcal{M}'_x\rangle)$.

For instance, for the noisy state shown in Eq.~\eqref{Eq:Snoisy}, one can effectively choose the corresponding witness in Eq.~\eqref{Eq:witnessTwo} to detect entanglement. Now the noise tolerance is the same as in the sole white noise case, i.e., $p=\frac{1}{3}$ as $N \rightarrow \infty$, no matter what value $\phi$ is. In addition, this protocol can further help to improve the experiment system by applying an correcting unitary according to the optimal $\phi_{opt}$.
\subsection{Efficient detection under noise model 2}\label{sSec:EEWnm2}
We further extend the witness $\mathcal{W}^2_{GHZ}$ in Eq.~\eqref{Eq:EWGHZtwo} to a family of witnesses, and have the following theorem.
\begin{theorem}\label{Th:2LM2}
The witness $\mathcal{W}_{\Psi_{\phi}^{\theta}}^2$ can detect entanglement near the state $\ket{\Psi_{\phi}^{\theta}}$,
\begin{equation}\label{Eq:2LM2}
\begin{aligned}
\mathcal{W}_{\Psi_{\phi}^{\theta}}^2=&\frac{2\max\{\cos^2\theta, \sin^2\theta\}+1}{4} \mathbb{I}-\mathcal{Z}\\
&-\frac{1}{4}\cos (2\theta) \mathcal{M}_z
-\frac{1}{4}\sin (2\theta)(\cos \phi \mathcal{M}_x + \sin \phi \mathcal{M}'_x),
\end{aligned}
\end{equation}
where $\mathcal{Z}$ is defined in Eq.~\eqref{Eq:noiseD}, $\mathcal{M}_z=\sigma_z\otimes \mathbb{I}^{\otimes  N-1}$, $\mathcal{M}_x=\sigma_x ^{\otimes N}$, and $\mathcal{M}_x'=\sigma_y\otimes \sigma_x ^{\otimes N-1}$.
\end{theorem}
It is clear that in general $\ket{\Psi_{\phi}^{\theta}}$ can not be transformed from the standard GHZ state by local unitary operations.
Thus, we cannot prove Theorem \ref{Th:2LM2} with the approach used in Theorem \ref{Th:2LM1}, and the following proof is based on the generalized stabilizer formula.
\begin{proof}
The GHZ state is a stabilizer state that is uniquely determined by the following $N$ independent stabilizer operators,
\begin{equation}\label{Eq:StaGHZ}
\begin{aligned}
S_1=\sigma_x ^{\otimes N},S_{2}=\sigma_z^1\sigma_z^2, S_{3}=\sigma_z^2\sigma_z^3, \cdots, S_{N}=\sigma_z^{N-1}\sigma_z^N
\end{aligned}
\end{equation}
and the witness $\mathcal{W}^2_{GHZ}$ in Eq.~\eqref{Eq:EWGHZtwo} can be equivalently written as,
\begin{equation}\label{Eq:GHZstab}
\begin{aligned}
2\mathcal{W}^2_{GHZ}=\frac{3}{2} \mathbb{I}-(P_1+P_2)
\end{aligned}
\end{equation}
with $P_1$ and $P_2$ two projectors determined by the stabilizers,
\begin{equation}\label{Eq:ProGHZ}
\begin{aligned}
P_1&=\frac{\mathbb{I}+S_1}{2}=\frac{\mathbb{I}+\sigma_x ^{\otimes N}}{2},\\
P_2&=\prod_{i=2}^{N}\frac{\mathbb{I}+S_i}{2}\equiv2\mathcal{Z}.
\end{aligned}
\end{equation}

Due to the fact $(\mathbb{I}-P_1)(\mathbb{I}-P_2)\geq 0$, one has
\begin{equation}\label{Eq:GHZproj}
\begin{aligned}
\ket{GHZ}\bra{GHZ}=P_1P_2\geq (P_1+P_2-\mathbb{I}),
\end{aligned}
\end{equation}
As a result, the witness $\mathcal{W}_{GHZ}^2$ is valid since $2\mathcal{W}^2_{GHZ} \geq \mathcal{W}_{GHZ}$ given in Eq.~\eqref{Eq:witnessGHZ} \cite{Toth2005Detecting}.

In the following, we construct the witness $\mathcal{W}_{\Psi_{\phi}^{\theta}}^2$ in Eq~\eqref{Eq:2LM2} by finding generalized stabilizers of $\ket{\Psi_{\phi}^{\theta}}$. Here ``generalized'' means that the stabilizer may be not in the Pauli tensor form.

It is not hard to see that the last $N-1$ stabilizers $S_2,S_3,\cdots S_N$ in Eq~\eqref{Eq:StaGHZ} also stabilize $\ket{\Psi_{\phi}^{\theta}}$. Thus, we only need to find the first updated one $S_1'$. The construction is based on the following fact: if $S$ stabilizes $\ket{\Psi}$, i.e., $S\ket{\Psi}=\ket{\Psi}$, $USU^{\dag}$ stabilizes $U\ket{\Psi}$. Note that $\ket{\Psi_{\phi}^{\theta}}$ can be prepared from the noisy circuit described in Fig.~\ref{Fig:GHZcir}.

Initially, the stabilizers of the product state $\ket{0}^{\otimes N}$ is $\sigma_z^1, \sigma_z^2,\cdots, \sigma_z^N$; after the single qubit unitary on the first qubit, $\sigma_z^1$ becomes
\begin{equation}\label{}
\begin{aligned}
U_1'\sigma_z^1U_1'^{\dag}=\cos(2\theta)\sigma_z^1 + \sin(2\theta)\left(\cos\phi\sigma_x^1+\sin\phi\sigma_y^1\right);
\end{aligned}
\end{equation}
finally, after the successive application of CNOT gates, the stabilizer shows,
\begin{equation}\label{Eq:Sta1new}
\begin{aligned}
S_1'=U_{CNOT}U_1'\sigma_z^1U_1'^{\dag}U_{CNOT}^{\dag}
\end{aligned}
\end{equation}
where $U_{CNOT}=[C_{N-1}X_N]\cdots [C_{2}X_3] [C_{1}X_2]$ is the CNOT gate sequence, with $C_iX_j$ denoting the CNOT gate on the qubit $j$ controlled by the qubit $i$. One can find that
\begin{equation}\label{Eq:Sta1new}
\begin{aligned}
S_1'=&\cos(2\theta)\sigma_z^1\otimes \mathbb{I}^{\otimes  N-1}+\\ &\sin(2\theta)\left(\cos\phi\sigma_x^{\otimes N}+\sin\phi\sigma_y^1\otimes\sigma_x^{\otimes N-1}\right)
\end{aligned}
\end{equation}
on account the following relations,
\begin{equation}\label{}
\begin{aligned}
&C_iX_j \sigma_x^i C_iX_j=\sigma_x^i\sigma_x^j,\\
&C_iX_j \sigma_y^i C_iX_j=\sigma_y^i\sigma_x^j,\\
&C_iX_j \sigma_z^i C_iX_j=\sigma_z^i.\\
\end{aligned}
\end{equation}

Similar as Eq.~\eqref{Eq:ProGHZ}, we can define two projectors $P_1'$ and $P_2'$ associated with the stabilizers of the state $\ket{\Psi_{\phi}^{\theta}}$, that is, $P_1'=\frac1{2}(\mathbb{I}+S_1')$ and $P_2'=P_2$. Then based on the witness $\mathcal{W}_{\Psi_{\phi}^{\theta}}$ in Eq.~\eqref{Eq:witnessF2}, we have the new witness with less LMSs,
\begin{equation}\label{Eq:2W2Proj}
\begin{aligned}
\mathcal{W}_{\Psi_{\phi}^{\theta}}^2&=\frac1{2}\left[(\max\{\cos^2\theta, \sin^2\theta\}+1) \mathbb{I}-(P_1'+P_2')\right]\\
&=\frac{2\max\{\cos^2\theta, \sin^2\theta\}+1}{4} \mathbb{I}-\mathcal{Z}-\frac{1}{4}S_1',
\end{aligned}
\end{equation}
with $S_1'$ given in Eq.~\eqref{Eq:Sta1new}. Similar as Eq.~\eqref{Eq:GHZproj}, one can verify $2\mathcal{W}_{\Psi_{\phi}^{\theta}}^2\geq \mathcal{W}_{\Psi_{\phi}^{\theta}}$ and the witness $\mathcal{W}_{\Psi_{\phi}^{\theta}}^2$ is valid.
\end{proof}

Similar as Sec.~\ref{sSec:EEWnm1}, one should find the minimal expectation value of all the witnesses in the family of Eq.~\eqref{Eq:2LM2}, i.e., $\min_{\phi,\theta} \mathrm{Tr}(\mathcal{W}_{\Psi_{\phi}^{\theta}}^2\rho_{pre})$, parameterized by $\phi$ and $\theta$. Equivalently,
\begin{equation}\label{Eq:2max2}
\begin{aligned}
\max_{\phi,\theta}&\left\{\cos (2\theta) \langle\mathcal{M}_z\rangle+\sin (2\theta)(\cos \phi \langle\mathcal{M}_x\rangle + \sin \phi \langle\mathcal{M}'_x\rangle)\right\}\\
&=\sqrt{\langle\mathcal{M}_z\rangle^2+\langle\mathcal{M}_x\rangle^2+\langle\mathcal{M}'_x\rangle^2}
\end{aligned}
\end{equation}
The optimal $\phi_{opt}$ to saturate the maximal value is given in Eq.~\eqref{Eq:optPhiTwo1}, and the optimal $\theta_{opt}$ satisfies
\begin{equation}\label{Eq:OpThetaTwo}
\begin{aligned}
\tan (2\theta_{opt})=\frac{\sqrt {\langle\mathcal{M}_x\rangle^2+\langle\mathcal{M}'_x\rangle^2}}{\langle\mathcal{M}_z\rangle},
\end{aligned}
\end{equation}
and $[\cos(2\theta_{opt}), \sin(2\theta_{opt})]$ is in the same quadrant with $\left(\langle\mathcal{M}_z\rangle, \sqrt {\langle\mathcal{M}_x\rangle^2+\langle\mathcal{M}'_x\rangle^2}\right)$.

Consequently, one can choose the optimal witness in the family of Eq.~\eqref{Eq:2LM2} to detect the entanglement, based on the associated optimal parameters $\phi_{opt}$ and $\theta_{opt}$. For instance, for the noisy state shown in Eq.~\eqref{Eq:Snoisy2}, the corresponding white noise tolerance shows,
\begin{equation}\label{Eq:twoTol2}
\begin{aligned}
p<\frac2{3}\min\{\cos^2\theta_{opt}, \sin^2\theta_{opt}\},
\end{aligned}
\end{equation}
for $N \rightarrow \infty$ no matter what value $\phi$ is, and the proof is left in Appendix \ref{Sec:AnoiseFinal}.

On the other hand, similar as the discussion at the end of Sec.~\ref{sSec:EWnm2}, one can also apply the detection protocol given in Sec.~\ref{sSec:EEWnm1} on the noise model 2 here. That is, using the optimal witness in the family of Eq.~\eqref{Eq:witnessTwo}. The optimization in Eq.~\eqref{Eq:2min} can determine the corresponding noise parameter $\phi_{opt}$ and thus the optimal witness. For the noisy state in Eq.~\eqref{Eq:Snoisy2}, the corresponding white noise tolerance reads,
\begin{equation}\label{Eq:NT1on2two}
\begin{aligned}
p<1-\frac{2}{\sin(2\theta)+2}.
\end{aligned}
\end{equation}
as $N\rightarrow\infty$. The proof is also left in Appendix \ref{Sec:AnoiseFinal}. Similar as the comparison at the end of Sec.~\ref{sSec:EWnm2}, the white noise tolerance in Eq.~\eqref{Eq:NT1on2two} is better than the one in Eq.~\eqref{Eq:twoTol2}.

The efficient detection protocol under the noise model 2 employs the same set of $3$ LMSs as the one in Sec.~\ref{sSec:EEWnm1}, but abstracts both noise parameters $\phi_{opt}$ and $\theta_{opt}$. This is because here we post-process measurement results more delicately. Even though the white noise tolerance of the corresponding witness in Eq.~\eqref{Eq:twoTol2} is not better than the one in Eq.~\eqref{Eq:NT1on2two}, the experiment system can be further improved with the noise parameters $\phi_{opt}$ and $\theta_{opt}$ extracted from the measurement results.
\section{conclusion and outlook}\label{Sec:conclude}
In this paper, by focusing on GHZ-like states, we propose an entanglement detection protocol to enhance the detection ability under some practical coherent noises, which only adds one LMS comparing to the original witness method. Our protocol can feedback the noisy parameters by postprocessing and further help to improve the experimental system.
The main idea behind the protocol is that we construct a set of measurements which can tomography all possible states affected by the coherent noise, and thus realize a family of entanglement witnesses. In addition, we further reduce the number of LMSs to 3, which makes the experimental realization more efficient.

There are a few prospective problems that can be explored in the future.
First, it is shown in the paper that even if one can obtain more parameters about the prepared state by delicate postprocessing, it may be not beneficial to the entanglement detection as shown by the noise tolerance comparison in Fig.~\ref{Fig:WNCompare}. Thus it is significant to investigate further whether it is a general phenomenon. Second, it is interesting to extend the current protocol to more general states, such as permutation-invariant states \cite{Toth2010Permutation,Zhou2019Decomposition} and stabilizer states \cite{gottesman1997stabilizer,Nielsen2011Quantum}, where quantum error correction or mitigation methods can be applied to eliminate or reduce the effect of coherent noises. Third, it is also intriguing to study the entanglement detection under other types of coherent noises, which appear in certain experimental systems. In addition, the detection of more detailed multipartite entanglement structures \cite{Huber2013Structure,Lu2018Structure,You2019graph} under coherent noises is significant to investigate.

\acknowledgments
We thank Chenghao Guo, Xiongfeng Ma and Qi Zhao for the insightful discussions. This work was supported by the National Natural Science Foundation of China Grants No.~11875173 and No.~11674193, and the National Key R\&D Program of China Grant No.~2017YFA0303900.

\appendix
\section{Derivation of miscellaneous noise tolerances in Eq.~\eqref{Eq:noiseTcom}, \eqref{Eq:noiseTcom2} , \eqref{Eq:Tol2} and \eqref{Eq:NT1on2}}\label{Sec:noiseT}
First, let us focus on Eq.~\eqref{Eq:noiseTcom} and \eqref{Eq:noiseTcom2}, which are noise tolerances of $\mathcal{W}_{GHZ}$ in Eq.~\eqref{Eq:witnessGHZ} for the noise model 1 and 2 respectively. Since the noisy state in Eq.~\eqref{Eq:Snoisy2} is more general than that in Eq.~\eqref{Eq:Snoisy}, we only show the derivation of Eq.~\eqref{Eq:noiseTcom2} here.

\begin{equation}\label{}
\begin{aligned}
\mathrm{Tr}(\mathcal{W}_{GHZ}\rho_{pre})&=\mathrm{Tr}\Bigg{\{}\left[\frac{1}{2} \mathbb{I}-\ket{GHZ}\bra{GHZ}\right]\\
&\ \ \ \ \ \ \ \ \
\left[(1-p)\ket{\Psi_{\phi}^{\theta}}\bra{\Psi_{\phi}^{\theta}}+p\frac{\mathbb{I}}{2^N}\right]\Bigg{\}}\\
&=\frac1{2}-(1-p)|\bra{GHZ}\Psi_{\phi}^{\theta}\rangle|^2-\frac{p}{2^N}\\
&=\frac1{2}-\frac{1-p}{2}[1+\sin(2\theta)\cos(\phi)]-\frac{p}{2^N}<0.
\end{aligned}
\end{equation}
As $N \rightarrow \infty$, one has
\begin{equation}\label{}
\begin{aligned}
\sin(2\theta)\cos(\phi)> \frac{p}{1-p}.
\end{aligned}
\end{equation}

Then, let us consider Eq.~\eqref{Eq:Tol2}, which is the tolerance of the optimal witness in Eq.~\eqref{Eq:witnessF2} for the noisy state in Eq.~\eqref{Eq:Snoisy2}. Since the fidelity optimization in Eq.~\eqref{Eq:Nmax2} helps us to determine the parameters $\phi$ and $\theta$ of the prepared state in Eq.~\eqref{Eq:Snoisy2}, one can choose the witness with the same parameters in Eq.~\eqref{Eq:witnessF2} and the noise tolerance shows,
\begin{equation}\label{}
\begin{aligned}
\mathrm{Tr}(\mathcal{W}_{\Psi_{\phi}^{\theta}}\rho_{pre})&=\mathrm{Tr}\Bigg{\{}\left[\max\{\cos^2\theta, \sin^2\theta\}\mathbb{I}-\ket{\Psi_{\phi}^{\theta}}\bra{\Psi_{\phi}^{\theta}}\right]\\
&\ \ \ \ \ \ \ \ \  \left[(1-p)\ket{\Psi_{\phi}^{\theta}}\bra{\Psi_{\phi}^{\theta}}+p\frac{\mathbb{I}}{2^N}\right]\Bigg{\}}\\
&=\max\{\cos^2\theta, \sin^2\theta\}-(1-p)-\frac{p}{2^N}<0.
\end{aligned}
\end{equation}
As $N \rightarrow \infty$, one has
\begin{equation}\label{}
\begin{aligned}
p<1-\max\{\cos^2\theta, \sin^2\theta\}=\min\{\cos^2\theta, \sin^2\theta\}.
\end{aligned}
\end{equation}

Finally, let us derive Eq.~\eqref{Eq:NT1on2}, which is the tolerance of the optimal witness in Eq.~\eqref{Eq:witnessF} for the noisy state in Eq.~\eqref{Eq:Snoisy2}. Since the fidelity optimization in Eq.~\eqref{Eq:Nmax} helps us to determine the parameter $\phi$ of the prepared state in Eq.~\eqref{Eq:Snoisy2}, one can choose the witness with the same parameter $\phi$ in Eq.~\eqref{Eq:witnessF} and the noise tolerance shows,
\begin{equation}\label{}
\begin{aligned}
\mathrm{Tr}(\mathcal{W}_{\Psi_{\phi}}\rho_{pre})&=\mathrm{Tr}\Bigg{\{}\left[\frac1{2} \mathbb{I}-\ket{\Psi_{\phi}}\bra{\Psi_{\phi}}\right]\\
&\ \ \ \ \ \ \ \ \ \left[(1-p)\ket{\Psi_{\phi}^{\theta}}\bra{\Psi_{\phi}^{\theta}}+p\frac{\mathbb{I}}{2^N}\right]\Bigg{\}}\\
&=\frac1{2}-(1-p)|\bra{\Psi_{\phi}}\Psi_{\phi}^{\theta}\rangle|^2-\frac{p}{2^N}\\&=\frac1{2}-\frac{1-p}{2}[1+\sin(2\theta)]-\frac{p}{2^N}<0.
\end{aligned}
\end{equation}
As $N \rightarrow \infty$, one has
\begin{equation}\label{}
\begin{aligned}
p<1-\frac{1}{1+\sin(2\theta)}.
\end{aligned}
\end{equation}
\section{Proof of decompositions in Eq.~\eqref{Eq:DecOff}}\label{prf:thm1}

First, note that the matrix form of $\cos \theta_k \sigma_x+ \sin \theta_k \sigma_y$ shows,
\begin{equation}
\begin{aligned}
 \cos \theta_k \sigma_x+ \sin \theta_k \sigma_y&=\left(
 \begin{array}{cc}
    0 & e^{-i\theta_k}\\
    e^{i\theta_k} & 0\\
  \end{array}
\right)\\
&=e^{-i\theta_k}\ket{0}\bra{1}+e^{i\theta_k}\ket{1}\bra{0},
\end{aligned}
\end{equation}

Let $l(b)=\sum_{i=0}^N b_i$ denote the weight of the binary string $b\in \{0,1\}^N$, and $\bar{b}$ is the bitwise inverse of $b$ with $\bar{b}_i=(b_i+1)\mod  2$. We can further rewrite the product operator $\mathcal{M}_{\theta_k}$ in Eq.~\eqref{Eq:LMSN} in the computational basis as follows.
\begin{equation}
\begin{aligned}
 \mathcal{M}_{\theta_k}=&\sum_b e^{i[l(b)-l(\bar{b})]\theta_k}\ket{b}\bra{\bar{b}}\\
=&\sum_b e^{i[2l(b)-N]\theta_k}\ket{b}\bra{\bar{b}}\\
=&\left(
 \begin{array}{ccc}
       & & e^{-iN\theta_k}        \\
         &\cdots&             \\
    e^{iN\theta_k} & & \\
  \end{array}
\right).
\end{aligned}
\end{equation}
Here in the second line we use the fact $l(b)+l(\bar{b})=N$ and $l(b)\in \{0,1,\cdots,N\}$.  Note that $\mathcal{M}_{\theta_k}$ only possesses terms on off-diagonal positions. For the clearness of the latter decomposition, we add a corresponding phase on $\mathcal{M}_{\theta_k}$ as,
\begin{equation}\label{Eq:AppLM}
\begin{aligned}
\mathcal{M}'_{\theta_k}&\equiv e^{iN\theta_k}\mathcal{M}_{\theta_k}=\sum_b e^{i2 \theta_k  l(b) }\ket{b}\bra{\bar{b}}\\
&=\sum_b e^{i\frac{2\pi k}{N+1} l(b)}\ket{b}\bra{\bar{b}}.
\end{aligned}
\end{equation}

From Eq.~\eqref{Eq:noiseOff} and \eqref{Eq:noiseOff1} in Main Text, $\mathcal{X}$ contains two terms $\mathcal{X}_+$and $\mathcal{X}_-$, and we rewrite them as,
\begin{equation}\label{}
\begin{aligned}
\mathcal{X}_+&=\frac1{2}\mathcal{X}_+',\\
\mathcal{X}_-&=\frac1{2i}\mathcal{X}_-',
\end{aligned}
\end{equation}
where $\mathcal{X}_+'$ and $\mathcal{X}_-'$ having matrix forms in the computational basis as,
\begin{equation}
\begin{aligned}
\mathcal{X}_+'=&\left(
 \begin{array}{ccccc}
      & & & & 1        \\
              & & &0 &                \\
       & & \cdots & &              \\
              & 0 & & &                    \\
   1 & & & &\\
  \end{array}
\right),\\
 \mathcal{X}_-'=&\left(
 \begin{array}{ccccc}
      & & & & 1        \\
              & & &0 &                \\
         & & \cdots & &               \\
              & 0 & & &                    \\
   -1 & & & &\\
  \end{array}
\right).
\end{aligned}
\end{equation}
Note that they show specific forms on the off-diagonal positions.

In the following, we derive the decomposition of $\mathcal{X}_+'$ and $\mathcal{X}_-'$ in terms of $\mathcal{M}_{\theta_k}'$ using discrete Fourier transform.
Note that $\mathcal{M}_{\theta_k}'$ shows the same coefficient on the terms $\ket{b}\bra{\bar{b}}$, if they share the same $l(b)$. Thus we only need to care about the weight of the binary $b$ and denote $t=l(b)$, which is the analog of the ``time'' domain, with $t=[0,1,\cdots,N]$. It is clear that
the function of $\mathcal{M}_{\theta_k}'$ on this domain is the Fourier basis function $e^{i\frac{2\pi k}{N+1} t}$, with the parameter $k$ being the analog of the ``frequency" domain. The corresponding functions of $\mathcal{X}_+'$ and $\mathcal{X}_-'$ on the time domain are $f_+(t)=[1,0\cdots,1]$ and $f_-(t)=[1,0\cdots,-1]$, respectively. By applying discrete Fourier transform, one has the coefficients showing
\begin{equation}
\begin{aligned}
F_+(k)=\frac1{N+1}\sum_{t=0}^N e^{-i\frac{2\pi k}{N+1}t}f_+(t)=\frac{1+e^{-i\frac{2\pi k N}{N+1}}}{N+1},\\
F_-(k)=\frac1{N+1}\sum_{t=0}^N e^{-i\frac{2\pi k}{N+1}t}f_-(t)=\frac{1-e^{-i\frac{2\pi k N}{N+1}}}{N+1}.
\end{aligned}
\end{equation}

Combing these coefficients with the operators, we have,
\begin{equation}\label{}
\begin{aligned}
\mathcal{X}_+&=\frac1{2}\sum_{k=0}^N  F_+(k) \mathcal{M}_{\theta_k}',\\
&= \frac1{2}\sum_{k=0}^N \frac{1+e^{-i\frac{2\pi k N}{N+1}}}{N+1} e^{iN\theta_k}\mathcal{M}_{\theta_k},\\
&=\frac1{N+1}\sum_{k=0}^N\cos(\frac{\pi k N}{N+1}) \mathcal{M}_{\theta_k},\\
&=\frac1{N+1}\sum_{k=0}^N (-1)^k \cos(\theta_k) \mathcal{M}_{\theta_k},
\end{aligned}
\end{equation}
where the last line is on account of $\frac{\pi kN}{N+1}=k\pi-\theta_k.$ Similarly,
\begin{equation}\label{}
\begin{aligned}
\mathcal{X}_-&=\frac1{2i}\sum_{k=0}^N  F_-(k) \mathcal{M}_{\theta_k}',\\
&= \frac1{2i}\sum_{k=0}^N \frac{1-e^{-i\frac{2\pi k N}{N+1}}}{N+1} e^{iN\theta_k}\mathcal{M}_{\theta_k},\\
&=\frac1{N+1}\sum_{k=0}^N\sin(\frac{\pi k N}{N+1}) \mathcal{M}_{\theta_k},\\
&=\frac{-1}{N+1}\sum_{k=0}^N (-1)^k \sin(\theta_k) \mathcal{M}_{\theta_k}.
\end{aligned}
\end{equation}
\section{Comparison between the white noise tolerances in Eq.~\eqref{Eq:Tol2} and~\eqref{Eq:NT1on2}}\label{Sec:CompareN}
Here, we compare the white noise tolerance in Eq.~\eqref{Eq:Tol2} using a family of witnesses $\mathcal{W}_{\Psi_{\phi}^{\theta}}$ with that using
a family of witnesses $\mathcal{W}_{\Psi_{\phi}}$ in Eq.~\eqref{Eq:witnessF}, for the noisy state in Eq.~\eqref{Eq:Snoisy2}.  In the following, we show the difference of them denoted by the function,
\begin{equation}\label{}
\begin{aligned}
g(\theta)=[1-\frac{1}{\sin(2\theta)+1}]-\min\{\cos^2\theta, \sin^2\theta\}\geq 0,
\end{aligned}
\end{equation}
where $\theta\in[0,\frac{\pi}{2}]$ and $g(\theta)=0$ as $\theta=0,\frac{\pi}{4}, \frac{\pi}{2}$. Note that $g(\theta)$ is symmetric with respective to $\theta=\frac{\pi}{4}$. Thus we only need to consider the regime $\theta\in[0,\frac{\pi}{4}]$,
\begin{equation}\label{}
\begin{aligned}
g(\theta)=&[1-\frac{1}{\sin(2\theta)+1}]-\sin^2\theta \\
&=\cos^2\theta-\frac{1}{\sin(2\theta)+1}\\
&=\cos^2\theta-\frac{1}{(\cos\theta+\sin\theta)^2} \geq 0
\end{aligned}
\end{equation}
equivalently, $\cos\theta(\cos\theta+\sin\theta)-1\geq 0$, that is, $\cos\theta\sin\theta\geq \sin^2\theta\rightarrow1\geq \tan \theta$ or $\sin\theta=0$. This true since $\theta\in[0,\frac{\pi}{4}]$.

\section{Derivation of noise tolerances in Eq.~\eqref{Eq:twoTol2}, \eqref{Eq:NT1on2two} and the comparison}\label{Sec:AnoiseFinal}
First, let us derive the noise tolerances in Eq.~\eqref{Eq:twoTol2}. Hereafter, we use $\phi$ and $\theta$ to denote $\phi_{opt}$ and $\theta_{opt}$ without confusion, and also denote $f(\theta)=\max\{\cos^2\theta, \sin^2\theta\}$ in $\mathcal{W}_{\Psi_{\phi}^{\theta}}^2$ for simplicity.

Note that in Eq.~\eqref{Eq:2W2Proj}, $\mathcal{W}_{\Psi_{\phi}^{\theta}}^2$ is written in the following form,
\begin{equation}\label{Eq:appEWP12}
\begin{aligned}
2\mathcal{W}_{\Psi_{\phi}^{\theta}}^2&=(f(\theta)+1) \mathbb{I}-(P_1'+P_2'),
\end{aligned}
\end{equation}
where $P_1'$ and $P_2'$ are two projectors determined by the stabilizers of the state $\ket{\Psi_{\phi}^{\theta}}$.

The noise tolerance is determined by $\mathrm{Tr}(\mathcal{W}_{\Psi_{\phi}^{\theta}}^2\rho_{pre})<0$. Inserting the witness of Eq.~\eqref{Eq:appEWP12} and the noisy state $\rho_{pre}$ of Eq.~\eqref{Eq:Snoisy2}, one has
\begin{equation}\label{Eq:AppNT}
\begin{aligned}
&\mathrm{Tr}\left\{\left[(f(\theta)+1) \mathbb{I}-(P_1'+P_2')\right]\rho_{pre}\right\}\\
&=(f(\theta)+1)-\mathrm{Tr}\left\{(P_1'+P_2')\left[(1-p)\ket{\Psi_{\phi}^{\theta}}\bra{\Psi_{\phi}^{\theta}}+p\frac{\mathbb{I}}{2^N}\right]\right\}\\
&=(f(\theta)+1)-2(1-p)-\frac{p}{2^N}\mathrm{Tr}( P_1'+ P_2'  )\\
&=(f(\theta)+1)-2(1-p)-\frac{p(2^{N-1}+2)}{2^N}<0,\\
\end{aligned}
\end{equation}
where in the third line we use the fact that $P_1'$ and $P_2'$ stabilize $\ket{\Psi_{\phi}^{\theta}}$, and in the final line $\mathrm{Tr}(P_1')=2^{N-1}$, $\mathrm{Tr}(P_2')=2$. From Eq.~\eqref{Eq:AppNT}, it is not hard to see that
\begin{equation}\label{}
\begin{aligned}
p=\frac{2^{N-1}}{3*2^{N-2}-1}(1-f(\theta))\simeq\frac{2}{3}\min\{\cos^2\theta, \sin^2\theta\},
\end{aligned}
\end{equation}
as $N \rightarrow \infty$.

Second, let us derive Eq.~\eqref{Eq:NT1on2two}, which is the tolerance of the optimal witness in Eq.~\eqref{Eq:witnessTwo} for the noisy state in Eq.~\eqref{Eq:Snoisy2}. Since the optimization in Eq.~\eqref{Eq:2min} helps to determine the parameter $\phi$ of the prepared state in Eq.~\eqref{Eq:Snoisy2}, one can choose the witness with the same parameter $\phi$ in Eq.~\eqref{Eq:witnessTwo}. Consequently, the white noise tolerance is the same with the noisy state,
\begin{equation}\label{Eq:Snoisy3}
\begin{aligned}
\rho_{pre}=(1-p)\ket{\Psi^{\theta}}\bra{\Psi^{\theta}}+p\frac{\mathbb{I}}{2^N},
\end{aligned}
\end{equation}
where $\ket{\Psi^{\theta}}=\cos\theta\ket{0}^{\otimes N}+\sin\theta \ket{1}^{\otimes N}$, under the detection of the witness $\mathcal{W}^2_{GHZ}$ in Eq.~\eqref{Eq:EWGHZtwo}.
\begin{equation}\label{Eq:AppNT1}
\begin{aligned}
&\mathrm{Tr}(\mathcal{W}^2_{GHZ}\rho_{pre})\\
&=\mathrm{Tr}\Bigg{\{}\left[\frac3{2} \mathbb{I}-(P_1+P_2)\right]\left[(1-p)\ket{\Psi^{\theta}}\bra{\Psi^{\theta}}+p\frac{\mathbb{I}}{2^N}\right]\Bigg{\}}\\
&=\frac3{2}-(1-p)[\frac3{2}+\frac1{2}\sin(2\theta)]-\frac{p}{2^N}\mathrm{Tr}( P_1+ P_2 )\\
&=\frac3{2}-(1-p)[\frac3{2}+\frac1{2}\sin(2\theta)]-\frac{p(2^{N-1}+2)}{2^N}<0.
\end{aligned}
\end{equation}
Here in the second line we applies the formula of $\mathcal{W}^2_{GHZ}$ in Eq.~\eqref{Eq:GHZstab}. The second line is due to the fact that $P_2$ stabilizes $\ket{\Psi^{\theta}}$ and $\mathrm{Tr}(P_1\ket{\Psi^{\theta}}\bra{\Psi^{\theta}})=\frac1{2}+\frac1{2}\sin(2\theta)$, and the final line $\mathrm{Tr}(P_1)=2^{N-1}$, $\mathrm{Tr}(P_2)=2$. From Eq.~\eqref{Eq:AppNT1}, it is not hard to see that
\begin{equation}\label{}
\begin{aligned}
p<\frac{\sin(2\theta)}{\sin(2\theta)+2-2^{2-N}}\simeq1-\frac{2}{\sin(2\theta)+2},
\end{aligned}
\end{equation}
as $N \rightarrow \infty$.

Finally, let us compare the two noise tolerances in Eq.~\eqref{Eq:twoTol2} and \eqref{Eq:NT1on2two}. Similar as Appendix \ref{Sec:CompareN}, we define the function $l(\theta)$ as the subtraction and consider the regime $\theta\in(0,\frac{\pi}{4})$ due to the symmetry,
\begin{equation}\label{}
\begin{aligned}
l(\theta)=&[1-\frac{2}{\sin(2\theta)+2}]-\frac2{3}\sin^2\theta \\
&=\frac2{3}\left(\frac1{2}+\cos^2\theta-\frac{3}{\sin(2\theta)+2}\right) \geq 0.
\end{aligned}
\end{equation}
That is, $(\frac1{2}+\cos^2\theta)(\sin(2\theta)+2)\geq 3$. After simplification, it is equivalent to $2+\cos(2\theta)\geq 2\tan(\theta)$, which is right for $\theta\in(0,\frac{\pi}{4})$.

%

\end{document}